\documentclass[runningheads,envcountsame,envcountsect]{llncs}

\usepackage{graphicx}
\usepackage{algorithmic}
\usepackage{algorithm}
\usepackage{amsmath}
\usepackage{amssymb}
\usepackage{amsfonts}
\usepackage{amsxtra}
\usepackage{enumitem}
\usepackage{newtxtext}
\usepackage{newtxmath}
\usepackage{mathtools}
\usepackage{etoolbox}
\usepackage{booktabs}
\usepackage{adjustbox}
\usepackage{siunitx}
\usepackage{comment}

\usepackage[round]{natbib}

\spnewtheorem{fact}[theorem]{Fact}{\bfseries}{\itshape}
\numberwithin{theorem}{section}
\numberwithin{equation}{section}

\makeatletter
\let\c@algorithm\c@theorem

\makeatother

\usepackage[svgnames]{xcolor}
\definecolor{darkgreen}{rgb}{0,.35,0}
\definecolor{darkblue}{rgb}{0,0,.1}
\definecolor{darkred}{rgb}{.6,0,0}

\makeatletter
\newcommand{\papertitle}{\@title}
\newcommand{\paperauthor}{\@author}
\makeatother

\makeatletter
\newcommand{\ReDeclareMathOperator}{%
  \@ifstar{\@declmathop\@empty}{\@declmathop o}}
\long\def\@declmathop#1#2#3{%
  \DeclareRobustCommand{#2}{\qopname\newmcodes@#1{#3}}}
\makeatother

\newcommand{\QQ}{\mathbb{Q}}
\DeclareRobustCommand{\QQbar}{%
  \mathord{%
    \ooalign{%
      $\mathbb{Q}$\cr
      \hidewidth\raisebox{1.7ex}{\rule{0.6em}{0.6pt}}\hidewidth\cr
    }%
  }%
}

\DeclareTextFontCommand{\boldemph}{\bfseries\em}

\newcommand{\CC}{\mathbb{C}}
\newcommand{\C}{\mathsf{C}}
\newcommand{\FF}{\mathbb{F}}
\newcommand{\K}{\mathsf{K}}
\newcommand{\F}{\mathsf{F}}
\newcommand{\E}{\mathsf{E}}
\newcommand{\Lgen}{\mathsf{L}_{\textrm{gen}}}
\newcommand{\Agen}{\mathsf{A}_{\textrm{gen}}}
\newcommand{\R}{\mathsf{R}}

\newcommand{\ZZ}{\mathbb{Z}}
\newcommand{\OK}{\mathcal{O}_\K}

\newcommand{\OLa}{\mathcal{O}_{\L_a}}
\newcommand{\Tau}{\mathcal{T}}
\newcommand{\Chi}{\mathcal{X}}

\newcommand{\gcrd}{\operatorname{gcrd}}

\newcommand{\rrem}{\operatorname{rrem}}

\newcommand{\ind}{\operatorname{ind}}
\newcommand{\Br}{\operatorname{Br}}

\newcommand{\disc}{\operatorname{disc}}
\newcommand{\Norm}{\operatorname{Norm}}
\newcommand{\Gal}{\operatorname{Gal}}
\newcommand{\q}{{\mathfrak{q}}}
\newcommand{\p}{{\mathfrak{p}}}

\newcommand{\Ralg}{{\widetilde\R}}
\newcommand{\Calg}{{\widetilde\C}}
\newcommand{\Falg}{{\widetilde\F}}

\newcommand{\tightdisplay}[1]{%
  \begingroup
  \setlength{\abovedisplayskip}{4pt}%
  \setlength{\belowdisplayskip}{4pt}%
  \setlength{\abovedisplayshortskip}{2pt}%
  \setlength{\belowdisplayshortskip}{2pt}%
  \[
  #1
  \]
  \endgroup\ignorespaces
}

\DeclareMathOperator{\sp}{sp}

\usepackage{hyperref}

\hypersetup{
  colorlinks=true,
  linkcolor=blue!50!black,
  citecolor=blue!50!black,
  urlcolor=blue!60!black,
  pdftitle={On Factoring Quantum-Plane Skew Polynomials over \texorpdfstring{$\QQ(\omega)(t)$}{Q(omega)(t)}},
  pdfauthor={Mark Giesbrecht}
}


\let\L\relax
\newcommand{\L}{\mathsf{L}}

\begin{document}

\title{On Factoring Quantum-Plane Skew Polynomials
  over~\texorpdfstring{$\QQ(\omega)(t)$}{Q(omega)(t)}}
\author{Mark Giesbrecht}
\institute{Cheriton School of Computer Science, Faculty of Mathematics, University of Waterloo, Canada
\email{mwg@uwaterloo.ca}}
\date{\today}

\maketitle

\begin{abstract}
  We study algorithms for factorization in the quantum plane of
  (dilation) skew polynomials over a function field of a cyclotomic
  field:
  \tightdisplay{
    \R=\K(t)[x;\sigma], \qquad \K=\QQ(\omega),
    \qquad \sigma(t)=\omega t,
  }
  where $\omega\in\CC$ is a primitive $m$-th root of unity. We start
  with the established approach through central elements and factor
  the central left multiples, staying in characteristic zero, to
  obtain a partial decomposition.  A two-level modular approach is
  proposed: specialize a central parameter to good algebraic values,
  study the resulting cyclic algebras over number fields, and then
  reduce further at good inert primes so that fast finite-field
  skew-factorization algorithms apply.  A prototype SageMath
  implementation is provided to experiment with the algorithms.  We
  then look at the effect of extending the field of constants from
  $\QQ(\omega)$ to $\QQbar$, an algebraic closure of $\QQ$, and
  factoring over $\QQbar(t)[x;\sigma]$. In this case we show
  factorization is decidable in the exact algebraic model based on
  finite extensions. 
\end{abstract}

\footnotetext{To appear,
  \href{https://www.casc-conference.org/index.html}{Computer Algebra
    in Scientific Computation (CASC)} conference, August
  31--September 4, 2026, Bath, UK}

\section{Introduction}

This paper studies factorization of skew polynomials over a function
field over a cyclotomic field with a dilation automorphism. This is a
characteristic-zero analogue of the Ore polynomial setting over
$\FF_q(t)$ that underlies earlier skew factorization algorithms. In
the root-of-unity quantum-plane case, the large univariate centre
makes a ``bound first'' approach effective.

Let $\K=\QQ(\omega)$, where $\omega\in\CC$ is a primitive $m$-th root
of unity, and let
\[
  \R=\K(t)[x;\sigma], \qquad xa=\sigma(a)x \ \ (\mbox{for any } a\in
  \K(t)), \qquad \sigma(t)=\omega t.
\]
Because $\sigma$ has finite order $m$, its fixed field is $\K(t^m)$
and the centre of $\R$ is the commutative principal ideal domain
\[
  \C=\K(t^m)[x^m]=\K(\Tau)[\Chi], \qquad \Tau=t^m, \ \Chi=x^m.
\]
The centre is a polynomial ring in the single central variable
$\Chi$ over the rational function field $\K(\Tau)$.

For our algorithms the initial input should first be cleared of
denominators, but the factorization stages are carried out for a monic
associate.  Let $F_{\rm in}\in \R\setminus\{0\}$ be the given input.
Clearing denominators and removing scalar content chooses an integral
representative $F\in \K[t][x;\sigma]$ and a unit
$\gamma\in\K(t)^\times$ with $F=\gamma F_{\rm in}$.  Put
$\ell=\operatorname{lc}_x(F)$ and $f=F^\sharp:=\ell^{-1}F$.  Then $f$
is monic and $F_{\rm in}=\gamma^{-1}\ell f$.  The scalar
$\gamma^{-1}\ell$ is a unit of $\R$, so $F_{\rm in}$ and $f$ have the
same right divisors and the same irreducibility status, up to
multiplication by units.  Below, unless explicitly stated otherwise,
$f$ denotes this working monic associate. A factorization
\tightdisplay{
  f=f_1f_2\cdots f_k
}
into irreducibles in $\R$ gives the corresponding factorization
\tightdisplay{
  F_{\rm in}=\gamma^{-1}\ell f_1f_2\cdots f_k .
}
Denominators in $t$ may necessarily occur in the factors even when the
chosen integral representative $F$ has coefficients in $\K[t]$.

We use right-divisibility conventions throughout: $h$ is a right
divisor of $g$ if $g=qh$.  We write $\gcrd(a,b)$ for the monic
greatest common right divisor, equivalently the monic generator of
$\R a+\R b$.  For $h\ne 0$, $\rrem(g,h)$ denotes the right remainder
in the division $g=qh+r$, where $r=0$ or $\deg_x r<\deg_x h$.

A nonzero $f\in\R$ is \boldemph{bounded} if $\R f\cap\C\ne 0$, i.e.,
if $f$ has a nonzero central left multiple.  This is automatic in
$\R$, since it is free over $\C$ and is a domain.  Following
\cite{Jac43}, the \boldemph{bound} of $f$ is the unique monic
$\varphi\in\C\cap\R f$ of minimal degree in $\Chi$, say $\varphi=uf$
with $u\in\R$.  Such a bound is unique and computable by linear
algebra.

The overall algorithmic approach initially follows that of
\cite{Gie98} and subsequent algorithms. First we compute a monic bound
$\varphi\in\C$. If $\varphi$ factors in $\C$, then right gcd computations
split $f$ into a rough factorization whose pieces have irreducible
central bounds; the multiplicities of central factors are handled
recursively. The difficult case is therefore a factor whose
bound is central and irreducible. This is the same point at which the
general theory over $\FF_q(t)$ passes from commutative central
factorization to simple algebras and zero divisors
\citep{GieZha03,GomLob19}. The same difficulty appears here
over $\K(t)$.

Our proposed method is a two-level modular approach. We proceed in
characteristic zero long enough to exploit the centre, the reduced
norm, and the cyclic-algebra structure attached to an irreducible
bound. We then specialize the single central parameter $\Tau=t^m$ to a
good algebraic value $a$, so that the problem descends to a cyclic
algebra over a number field. Finally, we reduce further at good inert
primes, factor the resulting finite-field images using the algorithms
of \cite{Gie98} or successors, such as \cite{CarLeb17}, and lift
factors back first to the specialized number field and then
to $\K(t)$.

Over $\K(t)$ we exhibit a central decomposition
algorithm, two inexpensive irreducibility filters, and a two-level
modular search for the remaining irreducible-bound case.
The search may return failure, but every factor it returns is verified
by exact division in the original ring.

When we enlarge the constants from $\K$ to $\QQbar$, we obtain
stronger results, albeit for a different factorization problem whose
output may be much larger. In this setting over an algebraically
closed field, the irreducible-bound algebra splits, by Tsen's theorem,
and full factorization becomes decidable in the algebraic model based
on finite extensions and root adjunction. We briefly discuss the
corresponding factorization problem in the Blum--Shub--Smale (BSS)
model over $\CC$, more precisely in BSS augmented with a univariate
root-choice primitive.

Section \ref{sec:background} surveys previous relevant work. Sections
\ref{sec:qplanes}--\ref{sec:alg} develop our proposed algorithm over
$\K(t)$. Section \ref{sec:complexity} discusses complexity and
practical remarks, Section \ref{sec:history} looks at the history and
related literature, while Section \ref{sec:algstatus} identifies the
parts that are currently algorithmic over $\K(t)$. Section
\ref{sec:implementation} describes our prototype implementation in
SageMath.  In Section \ref{sec:algclosure} we prove decidability over
$\QQbar(t)$ and explain the corresponding BSS interpretation.

\section{Background and previous work}
\label{sec:background}

\cite{Ore33} introduced skew polynomial rings and their basic
algebraic structure. Over finite coefficient fields, \cite{Gie98} gave
the first polynomial-time factorization algorithm for
$\FF_q[x;\sigma]$, and \citet{CarLeb17} later developed algorithms
with improved complexity and an implementation. Over $\FF_q(t)$,
\citet{GieZha03} established the centre/eigenring viewpoint that
underlies the present paper, although the general irreducible-bound
case requires subroutines whose general availability was
overstated. This issue is identified by \citet[Sec.~6.2]{GomLob19}: once
the bound is central and irreducible, we have to find zero divisors in
simple algebras over finite extensions of the centre. 

In the present paper we consider the univariate quantum plane over a
rational function field at a root of unity. This is exactly the
setting where $x^m$ and $t^m$ are central, so the centre is a
univariate polynomial ring over $\K(\Tau)$. Low-degree factorization
in quantum planes was studied by \citet{CouPri06}; related
irreducibility criteria for skew polynomials of degree at most four,
with applications to quantum planes and quantized Weyl algebras, were
given by \citet{BroPum21}. We are not aware of a general factorization
algorithm in the root-of-unity quantum-plane setting considered~here.

\section{Quantum planes, bounds, and factor degrees}
\label{sec:qplanes}

We begin with the structural results that govern every monic element
$f\in\R$ whose monic bound is an irreducible central polynomial
$\pi\in\C$.

The structure theory gives the degree pattern in this
irreducible-bound case; see \cite{GieZha03} and \cite{GomLob19}.

\enlargethispage{25pt}
\begin{fact}
  \label{fact:degrees}
  Assume that the monic bound of $f\in\R$ is an irreducible polynomial
  $\pi\in \C$ with $\pi\neq \Chi=x^m$. Set $s=\deg_\Chi(\pi)$,
  $\F=\C/(\pi)$, and $A=\R/\R\pi$.  Since $\pi$ is central,
  $\R\pi=\pi\R$ is a two-sided ideal.  The algebra $A$ is a central
  simple $\F$-algebra of degree $m$. Write $A\cong M_r(D)$, with $D$ a
  central division $\F$-algebra of degree $d=\ind(A)$; thus $rd=m$.
  Since any nonconstant right factor of $f$ has monic bound dividing
  $\pi$, and $\pi$ is irreducible, every such factor also has bound
  $\pi$.  Since $\deg_x\pi=ms$, every irreducible factor of $f$ in
  $\R$ has degree $e=sd$ in $x$.  In particular, every complete
  factorization of $f$ has the form \tightdisplay{ f=h_1h_2\cdots
    h_\ell, \qquad \deg_x h_i = s d \quad (1\leq i\leq \ell), \qquad
    \ell=(\deg_x f)/(sd), } and the only possible irreducible factors
  are of degree $sd$, with $d\mid m$.
\end{fact}

Note that the exceptional central prime $\Chi=x^m$ is elementary: $x$
is a normal irreducible element, and if $x$ is a right factor of $f$,
then we first strip the maximal right power of $x$.

The following inexpensive irreducibility criterion is an immediate
consequence.
\begin{corollary}
  \label{cor:maxdeg}
  Let $f\in \R$ be monic and assume that its monic bound is an
  irreducible polynomial $\pi\in\C$ with $\pi\neq\Chi$. If
  $\deg_x(\pi)=m\,\deg_x(f)$, then $f$ is irreducible.
\end{corollary}

\section{Central decomposition}
\label{sec:centraldecomp}

We first separate the part of the input that is already controlled by
the centre from the genuinely irreducible-bound part. The normal
factor $x$ is handled explicitly before the central recursive step.

\begin{algorithm}[H]
\caption{Central decomposition stage}
\label{alg:central}
\textbf{Input:} a monic nonzero polynomial $f\in \R$.

\textbf{Output:} a rough product decomposition
$f=g_1g_2\cdots g_\ell$, in which each $g_j$ is either the irreducible
factor $x$ or is monic with irreducible monic bound $\pi_j\in\C$,
$\pi_j\neq\Chi$.

\begin{enumerate}[label=\textup{(\arabic*)},leftmargin=2.2em]
\item Strip the maximal right power of $x$: write $f=f_0x^\varepsilon$, with
  $\varepsilon\geq 0$ and either $f_0=1$ or $x\nmid_r f_0$. If $f_0=1$, return
  $\varepsilon$ copies of $x$.
\item Compute the monic bound $\varphi\in \C$ of $f_0$ in
  $\R=\K(t)[x;\sigma]$, for example by linear algebra or by the
  algorithm of \citet{GomLob19}. Since $x\nmid_r f_0$, the bound
  $\varphi$ is not divisible by $\Chi$.
\item If $\varphi$ is irreducible in $\C$, return the one-term
  factorization $f_0=g_1$, followed by $\varepsilon$ copies of $x$.
\item Otherwise choose a monic irreducible factor $\pi\neq\Chi$ of
  $\varphi$ and write $ \varphi=\rho\pi$, with $\rho,\pi\in\C$ monic
  and nonconstant.
\item Compute the monic right gcd in $\R$: $v:=\gcrd(f_0,\pi)$.
\item Compute the right quotient $u\in\R$ satisfying $f_0=uv$.
\item Apply the present algorithm recursively to $u$. If the recursive
  call returns
  \tightdisplay{
    u=u_1u_2\cdots u_r,
    \hspace*{15pt}\mbox{then return}\hspace*{15pt}
    f=u_1u_2\cdots u_r v x^\varepsilon.
  }
\end{enumerate}
\end{algorithm}

Algorithm~\ref{alg:central} is the standard recursive rough-factor
stage of the bound method: at each step it strips off a factor
attached to a single irreducible central divisor and then recurses on
the cofactor. The word ``rough'' is important: the output pieces
have irreducible central bounds, but they need not be irreducible
skew polynomials.

\begin{proposition}
  \label{prop:centralcorrect}
  Algorithm~\ref{alg:central} terminates. Apart from the explicitly
  returned copies of $x$, each returned factor is monic and has an
  irreducible monic bound different from $\Chi$.
\end{proposition}

\begin{proof}[Sketch]
  If $\gcrd(f_0,\pi)=1$, a left B\'ezout identity shows that $\rho$ is
  a smaller central left multiple of $f_0$, contradicting minimality
  of $\varphi$; hence $v=\gcrd(f_0,\pi)$ is nonconstant. If $v=f_0$,
  then $\pi$ is itself a smaller central left multiple of $f_0$, again
  impossible. Thus the cofactor $u$ has smaller positive degree in
  $x$. Since $\pi$ is a central left multiple of $v$, the monic bound
  of $v$ divides $\pi$, hence equals $\pi$ because $\pi$ is
  irreducible. The case $\pi=\Chi$ is excluded by the initial
  stripping step. Finally, if $\varphi$ were divisible by $\Chi$, then
  $\varphi=w f_0$ would give $w f_0=0$ in $\R/\R\Chi$. Since
  $x\nmid_r f_0$, the image of $f_0$ is invertible in this
  skew-polynomial quotient, so $w\equiv0\pmod{\R\Chi}$ and
  $\varphi/\Chi$ would be a smaller central left multiple of $f_0$.
  Therefore the recursion proceeds on smaller $x$-degree and preserves
  the product identity by right division. \qed
\end{proof}

This approach to the computation just described is well established; see
\citep{GieZha03,GomLob19}. The hard input class is therefore when
$g\in\R$ has an irreducible monic bound $\pi\in\C$.  The remainder of
this article will focus on this case.

\subsection{Norms and cyclic algebras}
\label{subsec:norms}

Our algorithm above, and the methods described below, are based on
central bounds, their factorizations, and the remaining
irreducible-bound case. Alternatively, \emph{reduced norms} give
another commutative invariant, and have been used algorithmically in
skew factorization over finite fields and in the finite-order setting
\citep{CarLeb17, PumTho22}. This norm is another central invariant
associated with the same central prime as the irreducible-bound case.
For fixed $m$ it is relatively inexpensive to compute, and sometimes
gives a useful preliminary test before the specialization stage.

Let $\E$ be the field of fractions $\K(\Tau)(\Chi)$ of $\C$ and
$\Lgen:=\K(t)(\Chi)$.  Localizing $\R$ at the nonzero central elements
gives
\[
  \Agen:=\E\otimes_{\C}\R \cong (\Lgen/\E,\sigma,\Chi),
\]
the generic cyclic algebra of degree $m$. Here $\Lgen/\E$ is the
cyclic extension obtained from $t^m=\Tau$, the automorphism is induced
by $\sigma(t)=\omega t$, and the crossed-product relation is
$x^m=\Chi$.

For $g\in\R$, collect its image in $\Agen$ modulo powers of $x^m=\Chi$
and write it on the right $\Lgen$-basis $1,x,\dots,x^{m-1}$:
\[
  g=a_0+x a_1+\cdots+x^{m-1}a_{m-1}, \qquad a_i\in \Lgen.
\]
Equivalently, a left-coefficient expression $\sum b_i x^i$ is
converted by using $b_i x^i=x^i\sigma^{-i}(b_i)$. Let $M_g$ be the
$m\times m$ matrix of left multiplication by $g$ on this right
$\Lgen$-basis. We define
\[
  \Norm(g):=\det(M_g).
\]
This determinant is the reduced norm of $g$ in $\Agen$, not the
determinant of left multiplication on the full $\E$-vector space
$\Agen$, which would be a power of the reduced norm. Although the
matrix entries lie in $\Lgen$, the determinant lies in $\E$; for
$g\in\R$, the finite-order skew-polynomial norm results of
\citet{PumTho22} imply that in fact
\[
  \Norm(g)\in \C=\K(\Tau)[\Chi].
\]
Moreover $\Norm(g)$ is a central left multiple of $g$. Thus, the monic
bound of $g$ divides the monic associate of $\Norm(g)$ in $\C$.
With classical linear algebra, after collecting coefficients modulo
$x^m=\Chi$, the determinant above can be computed using $O(m^3)$
arithmetic operations over $\Lgen$. 

The following sufficient test is a direct consequence of the
multiplicativity of the reduced norm; see \citep{PumTho22}.

\begin{proposition}
  \label{prop:normprefilter}
  Let $g\in\R$ be monic. If the monic associate of $\Norm(g)$ is
  irreducible in $\C$, then $g$ is irreducible in $\R$.
\end{proposition}

\begin{proof}
  If $g=qh$ with $q,h$ monic of positive $x$-degree, multiplicativity
  gives $\Norm(g)=\Norm(q)\Norm(h)$. For a monic skew polynomial of
  positive $x$-degree, the determinant formula has positive
  $\Chi$-degree; hence neither norm factor is a unit of $\C$. Thus an
  irreducible monic associate of $\Norm(g)$ rules out such a
  factorization. \qed
\end{proof}

Following Algorithm~\ref{alg:central}, it is useful to compute the
reduced norm on factors in the rough factorization (with irreducible
bounds).  If Proposition~\ref{prop:normprefilter} applies we may be
able to certify the factor irreducible.  This test is not necessary
for correctness, but can be useful in an implementation.

\section{Good characteristic-zero specializations $\Tau\mapsto a$}
\label{sec:goodspecs}

We now turn to the problem of factoring an irreducible-bound input
$g\in \R$ with monic bound
\[
  \pi\in\C, \qquad s=\deg_\Chi \pi, \qquad n_g=\deg_x g.
\]
Fix for the remainder of this section a left quotient $u\in \R$
with $\pi=ug$.  Our first modular level specializes the single central
parameter $\Tau=t^m$ to an algebraic integer $a\in \OK$, where $\OK$
is the ring of all algebraic integers in $\K$. To make specialization
explicit, write every coefficient of $g$ in the $\K(\Tau)$-basis
$1,t,\dots,t^{m-1}$:
\[
  \K(t)=\bigoplus_{j=0}^{m-1} \K(\Tau)t^j,
\]
so that
\[
  g = \sum_{i=0}^{n_g}
      \left(\sum_{j=0}^{m-1} c_{ij}(\Tau)t^j\right)x^i,
  \qquad c_{ij}(\Tau)\in \K(\Tau).
\]
The same applies to the coefficients of $u$ and of $\pi$. Clearing
finitely many denominators in these expansions gives a finite set of
bad values of $\Tau$ that must be avoided for specialization even to
be defined.

\begin{definition}
  \label{def:goodspec}
  An algebraic integer $a\in \OK$ is called a \boldemph{good central
    specialization for $(g,\pi)$}, or simply \boldemph{good for
    $(g,\pi)$}, if all of the following hold:
  \begin{enumerate}[label=\textup{(S\arabic*)},leftmargin=2.2em]
  \item no denominator occurring in the $\Tau$-coefficients of $g$,
    $u$, or $\pi$ vanishes at $\Tau=a$;
  \item the specialization of $\pi$ has the same $\Chi$-degree as
    $\pi$, and $\disc_\Chi(\pi(a,\Chi))\neq 0$;
  \item the specialized central polynomial
  \tightdisplay{
    \pi_a(\Chi):=\pi(a,\Chi)\in \K[\Chi]
  }
  is irreducible and satisfies $\pi_a\neq \Chi$;
  \item the Kummer polynomial $y^m-a\in \K[y]$ is irreducible.
\end{enumerate}
\end{definition}

Condition \textup{(S2)} is mostly automatic from monicity once a fixed
presentation has been cleared of denominators, but it is useful to
record explicitly because discriminants and leading coefficients will
also be excluded in later reductions. Condition \textup{(S4)} implies
that
\tightdisplay{
  \L_a=\K(\theta_a), \qquad \theta_a^m=a,
}
has degree $m$ over $\K$. Since $\K$ contains the group $\mu_m$ of
$m$-th roots of unity, the map
\tightdisplay{
  \sigma_a(\theta_a)=\omega\theta_a
}
extends to a $\K$-automorphism of $\L_a$ of order $m$, and its fixed
field is $\K$. Thus, the centre of $R_a:=\L_a[x;\sigma_a]$ is 
$Z(R_a)=\K[x^m]=\K[\Chi]$.

For the correctness statements below we only use that any chosen $a$
satisfying Definition~\ref{def:goodspec} behaves as claimed.  Let
$\Delta(\Tau)\in\K[\Tau]$ be a common denominator for all coefficients
of $g$, $u$, and $\pi$ when these are written in the $\K(\Tau)$-basis
$1,t,\dots,t^{m-1}$.  Set
\[
  B_\Delta
  =\K[\Tau,\Delta(\Tau)^{-1},t]/(t^m-\Tau)
  \subset \K(t).
\]
For $a$ satisfying \textup{(S1)}, define the specialization
homomorphism
\tightdisplay{
  \sp_a:B_\Delta\longrightarrow \L_a, \qquad
  \sum_{j=0}^{m-1} r_j(\Tau)t^j\longmapsto
  \sum_{j=0}^{m-1} r_j(a)\theta_a^j,
}
where each $r_j(\Tau)$ is regular at $\Tau=a$. Thus $\sp_a$ is not a
homomorphism from all of $\K(t)$, but from the localized coefficient
ring on which the chosen input data are defined.

Applying this map coefficientwise to $g$ and $u$ produces specialized
skew polynomials
\tightdisplay{
  g_a,u_a\in R_a=\L_a[x;\sigma_a]
}
satisfying $\pi_a=u_ag_a$.  Because $g$ is monic in $x$, the
specialization $g_a$ is also monic and nonconstant.

In practice we may sample rational integers $a\in\ZZ\subset\OK$, but
the definition of ``goodness'' uses exact tests. In particular,
irreducibility of $y^m-a$ is tested in $\K[y]$ using 
commutative factorization algorithms (see, e.g., \citep{GatGer13}).
The specialized irreducible-bound problem remains a cyclic-algebra
problem, but now over a number field.

\begin{proposition}
  \label{prop:specialized}
  Let $a$ be good for $(g,\pi)$. Then the monic bound of
  $g_a\in R_a=\L_a[x;\sigma_a]$ is $\pi_a$. Put
  \[
    F_a=\K[\Chi]/(\pi_a), \qquad A_a=R_a/R_a\pi_a .
  \]
  Then $A_a$ is a central simple $F_a$-algebra of degree $m$. If
  \[
    A_a\cong M_{r_a}(D_a),
  \]
  with $D_a$ a central division $F_a$-algebra of degree $d_a$, so
  $r_a d_a=m$, then every irreducible factor of $g_a$ in $R_a$ has
  $x$-degree $s d_a$.
\end{proposition}

\begin{proof}
  The identity $\pi=ug$ specializes to $\pi_a=u_ag_a$ in $R_a$. Since
  $\pi_a\in\K[\Chi]=Z(R_a)$, it is a central left multiple of
  $g_a$. Therefore the monic bound of $g_a$ divides $\pi_a$ in
  $\K[\Chi]$.

  The specialization $g_a$ is still monic in $x$, hence nonconstant.
  Thus, its bound is nonconstant. By \textup{(S3)}, $\pi_a$ is
  irreducible in $\K[\Chi]$ and $\pi_a\neq\Chi$, so the only possible
  nonconstant monic divisor of $\pi_a$ is $\pi_a$ itself. Hence the
  monic bound of $g_a$ is~$\pi_a$.

  Since $\pi_a$ is central, $R_a\pi_a=\pi_a R_a$ is a two-sided ideal.
  The ring $R_a$ has centre $\K[\Chi]$, and the image of $\Chi$ in
  $F_a$ is nonzero because $\pi_a\neq\Chi$; hence it is invertible.
  The same argument as in Fact~\ref{fact:degrees} shows that $A_a$ is
  a central simple $F_a$-algebra of degree $m$.

  Finally, \textup{(S2)} gives $\deg_\Chi(\pi_a)=\deg_\Chi(\pi)=s$.
  Applying the same argument as in Fact~\ref{fact:degrees} in the
  specialized ring $R_a$ yields that every irreducible factor of $g_a$
  has degree $s d_a$ in $x$. \qed
\end{proof}

The integer $d_a=\ind(A_a)$ is, in principle, decidable by known
cyclic-algebra methods over number fields. In particular,
\citet{Han07} reduces splitting and isomorphism questions for cyclic
algebras to field-theoretic tasks including norm equations.
Specializing $\Tau\mapsto a$ keeps the computation in characteristic
zero, but moves it to a number field where these methods apply.

\section{Good prime reductions of a specialization}
\label{sec:goodreds}

After a good characteristic-zero specialization, the second modular
level reduces the specialized problem to finite fields, where we have
provably effective and efficient algorithms.

We assume throughout this section that we have fixed a good
specialization $a\in\OK$ and the number field $\L_a/\K$.  For the
modular step we therefore choose primes for which the reduction of
$\theta_a$ still has degree $m$.

\begin{definition}
  \label{def:goodprime}
  A nonzero prime ideal $\q\subset \OK$ is called \boldemph{good
    for the specialization $a$} if all of the following hold:
  \begin{enumerate}[label=\textup{(P\arabic*)},leftmargin=2.4em]
  \item the residue characteristic of $\q$ is prime to $m$, and
    $\q$ is unramified and inert in $\L_a/\K$;
  \item $\q$ avoids the denominators, leading coefficients, and
    discriminants arising in $g_a$, $u_a$, and $\pi_a$;
  \item the reduction $\bar\pi_{a,\q}$ of $\pi_a$ modulo $\q$ is
    irreducible in $k_{\q}[\Chi]$ and satisfies
    $\bar\pi_{a,\q}\neq \Chi$, where $k_{\q}=\OK/\q$.
\end{enumerate}
\end{definition}

Correctness of the returned factors uses only primes satisfying
Definition~\ref{def:goodprime}. The search for such primes is
therefore part of the heuristic modular stage.  If no such primes are
found within the prescribed budget, the stage returns failure.

If $\q$ is inert, there is a unique prime $\p$ of $\L_a$ above it, and
the residue field $\lambda_{a,\q}=\OLa/\p$ has degree $m$ over
$k_{\q}$. Since the residue characteristic is prime to $m$, the
reduction of $\omega$ still has order $m$. The automorphism $\sigma_a$
preserves $\p$ and therefore descends to an automorphism
$\bar\sigma_{a,\q}\in \Gal(\lambda_{a,\q}/k_{\q})$ of order $m$. Thus
we obtain a finite-field skew polynomial ring
$\bar \R_{a,\q}=\lambda_{a,\q}[x;\bar\sigma_{a,\q}]$.  Let
$\bar g_{a,\q}$ denote the reduction of $g_a$.

\begin{proposition}
  \label{prop:splitmod}
  Let $a$ be good for $(g,\pi)$ and let $\q$ be good for $a$. Then:
  \begin{enumerate}[label=\textup{(\alph*)},leftmargin=2.2em]
  \item the monic bound of $\bar g_{a,\q}$ is the reduction
    $\bar\pi_{a,\q}$ of $\pi_a$;
  \item every irreducible factor of $\bar g_{a,\q}$ in
    $\bar \R_{a,\q}$ has degree $s$.
  \end{enumerate}
\end{proposition}

\begin{proof}
  Reducing the identity $\pi_a=u_ag_a$ modulo $\q$ gives
  $\bar\pi_{a,\q}=\bar u_{a,\q}\,\bar g_{a,\q}$.  Thus
  $\bar\pi_{a,\q}$ is a central multiple of $\bar g_{a,\q}$, so the
  monic bound of $\bar g_{a,\q}$ divides $\bar\pi_{a,\q}$ in
  $k_{\q}[\Chi]$. Since $\bar g_{a,\q}$ is still monic and
  nonconstant, its bound is nonconstant. Because $\bar\pi_{a,\q}$ is
  irreducible and $\bar\pi_{a,\q}\neq \Chi$ by \textup{(P3)}, it
  follows that the monic bound of $\bar g_{a,\q}$ is
  $\bar\pi_{a,\q}$. This proves \textup{(a)}.
  
  For \textup{(b)}, the same argument used in
  Proposition~\ref{prop:specialized} shows that the quotient
  $\bar A_{a,\q}=\bar \R_{a,\q}/\bar \R_{a,\q}\bar\pi_{a,\q}$ is a
  central simple algebra over the finite field
  $\bar F_{a,\q}:=k_{\q}[\Chi]/(\bar\pi_{a,\q})$.  Since
  $\Br(\bar F_{a,\q})=0$, this algebra is split, so its index is $1$.
  Fact~\ref{fact:degrees} in the specialized finite-field setting
  therefore gives factor degree $s$. \qed
\end{proof}

Proposition~\ref{prop:splitmod} explains why the modular step is useful.
At a good prime, the specialized quotient becomes a split algebra over a
finite field, where skew-polynomial factorization is efficient. The price
is that the modular factorization is usually finer than the desired
factorization in characteristic zero, so the modular factors must be
recombined and then verified in the original ring.

\subsection{Factoring a specialization by inert prime reductions}

For fixed $a$, we factor $g_a$ over $\L_a[x;\sigma_a]$ by repeatedly
reducing to $\bar \R_{a,\q}$, factoring there, and lifting blocks of degree $sd_a$.

For the algorithms below, a \emph{search budget} means a prescribed
finite range of data to be tried.  For a fixed specialization $a$ we
write $\mathcal B_a=(M_a,Q_a,C_a)$ for a finite inert-prime and
recombination budget: only prime ideals $\q\subset\OK$ with
$\Norm_{\K/\QQ}(\q)\le Q_a$ are tested, at most $M_a$ good inert
primes are used, and at most $C_a$ compatible modular block
collections are passed to Chinese remaindering and rational
reconstruction.  ``Exhausted'' in
Algorithm~\ref{alg:onespecialization} means exhausted within this
prescribed finite range; it does not mean exhaustion of all possible
good inert primes.

\begin{algorithm}[H]
\caption{Factoring one specialization}
\label{alg:onespecialization}
\textbf{Input:} a good specialization $a$ for an irreducible-bound
input $g$, the block size $d_a=\ind(A_a)$, and a finite budget
$\mathcal B_a=(M_a,Q_a,C_a)$.

\textbf{Output:} a factorization of $g_a$ in
$\L_a[x;\sigma_a]$, or failure after $\mathcal B_a$ is exhausted.

\begin{enumerate}[label=\textup{(\arabic*)},leftmargin=2.2em]
\item Search the prescribed prime range $\Norm_{\K/\QQ}(\q)\le Q_a$
  and choose at most $M_a$ good inert primes $\q_1,\dots,\q_r$ of $\K$
  for the specialization $a$.
\item For each $\q_i$:
  \begin{enumerate}[label=\textup{(\alph*)},leftmargin=2em]
    \item construct the finite field $\lambda_{a,\q_i}$ and the reduced skew ring $\bar \R_{a,\q_i}$;
    \item factor $\bar g_{a,\q_i}$ completely by \citet{CarLeb17} or
      \citet{Gie98};
    \item up to the cap $C_a$, try ordered recombinations of the
      degree-$s$ modular factors, and keep those that give candidate
      monic right divisors of degree $s d_a$;
  \end{enumerate}
\item For each collection of modular blocks that passes the matching
  tests, and after determining a common denominator for the candidate
  coefficients, use Chinese remaindering across the primes
  $\p_i\subset \OLa$ above $\q_i$ to reconstruct candidate coefficient
  data in an integral basis of~$\OLa$.
\item Recover a monic candidate factor $\widehat h_a\in \L_a[x;\sigma_a]$
  of degree $sd_a$ by rational reconstruction in $\L_a$, and certify
  it by right division of $g_a$.

\item If certification succeeds, then $\widehat h_a$ is irreducible:
  by Proposition~\ref{prop:specialized}, every irreducible factor of
  $g_a$ has $x$-degree $sd_a$, while
  $\deg_x(\widehat h_a)=sd_a$. Recurse on the cofactor, and
  continue until $g_a$ is completely factored.

\item If no candidate survives among the primes and block collections
  allowed by $\mathcal B_a$, return failure.
  
\end{enumerate}
\end{algorithm}

Algorithm~\ref{alg:onespecialization} is a bounded search with exact
checking. The finite-field factorizations are fast and exact, but the
search may miss the recombinations needed over $\L_a$. In particular,
Step~\textup{(2c)} only proposes candidate blocks from the modular
factors; it does not claim that one fixed finite-field factorization
lists all right divisors of degree $sd_a$. Every accepted specialized
factor is checked by division in the specialized ring.

\section{Lifting from several specializations back to $\K(t)$}
\label{sec:lift}

The finite-field factorizations obtained from several specializations
must then be lifted back to $\K(t)$.  Assume now that for several good
values $a_1,a_2,\dots,a_r\in \OK$ we have obtained specialized factors
of $g_{a_i}$ in $\L_{a_i}[x;\sigma_{a_i}]$.

Write a candidate specialized factor in the basis
$1,\theta_{a_i},\dots,\theta_{a_i}^{m-1}$:
\[
  h_{a_i}=\sum_{\ell=0}^{e}\left(\sum_{j=0}^{m-1} c_{\ell
      j}(a_i)\theta_{a_i}^j\right)x^{\ell}, \qquad c_{\ell j}(a_i)\in
  \K,
\]
where $e$ is the candidate degree over $\K(t)$. A candidate factor
over $\K(t)$ has the form
\[
  h=\sum_{\ell=0}^{e}\left(\sum_{j=0}^{m-1} c_{\ell
      j}(\Tau)t^j\right)x^{\ell}, \qquad \mbox{where\ \ } c_{\ell j}(\Tau)\in \K(\Tau).
\]
Thus, each good specialization supplies the values
\[
  c_{\ell j}(\Tau)\big|_{\Tau=a_i}=c_{\ell j}(a_i).
\]
This is now a \emph{commutative} interpolation problem over $\K$.

\begin{definition}
  \label{def:Bcompat}
  For fixed degree $e$, a family of specialized factors
  $h_{a_1},\ldots,h_{a_r}$ is called \boldemph{$B$-compatible} if, for
  every lower coefficient pair $(\ell,j)$, with $\ell<e$ and
  $0\leq j<m$, the data
  $(a_1,c_{\ell j}(a_1)),\dots,(a_r,c_{\ell j}(a_r))$ interpolate a
  rational function in $\K(\Tau)$ with numerator and denominator
  degree at most $B$.
\end{definition}

Here $B$ is the interpolation degree bound in $\Tau$: it bounds the
degrees of the numerator and denominator of each rational coefficient
function. It is not a coefficient-height bound.  For
Algorithm~\ref{alg:lift} we write
$\mathcal B_{\rm lift}=(r_{\max},B,C_{\rm loc},C_{\rm tup})$, where
$r_{\max}$ bounds the number of specializations used, $C_{\rm loc}$
bounds the number of local degree-$e$ candidates retained from each
specialized factorization, and $C_{\rm tup}$ bounds the number of
$B$-compatible tuples tested.

\begin{algorithm}[H]
\caption{Lifting from specializations}
\label{alg:lift}

\textbf{Input:} an irreducible-bound input $g$, specialized
factorizations of $g_{a_1},\dots,g_{a_r}$, and a finite lifting budget
$\mathcal B_{\rm lift}=(r_{\max},B,C_{\rm loc},C_{\rm tup})$.

\textbf{Output:} a factorization of $g$ in $\R$, or failure
after $\mathcal B_{\rm lift}$ is exhausted.

\begin{enumerate}[label=\textup{(\arabic*)},leftmargin=2.2em]
  \item For each divisor $d$ of $m$ with $sd\mid n_g$ and $sd<n_g$, set $e=sd$.

  \item For each of at most $r_{\max}$ retained specializations $a_i$,
    build a candidate list of at most $C_{\rm loc}$ monic degree-$e$
    right divisors of $g_{a_i}$ from selected products of the
    irreducible factors found in Algorithm~\ref{alg:onespecialization}
    (noting that a specialized divisor of degree $e$ must be a product
    of $e/(sd_{a_i})$ specialized irreducible factors, so we only use
    $a_i$ with $sd_{a_i}\mid e$).

\item Search, among the retained specializations for this $e$, through
  at most $C_{\rm tup}$ tuples of degree-$e$ divisors
  $(h_{a_1},\dots,h_{a_r})$ that are $B$-compatible in the sense of
  Definition~\ref{def:Bcompat}.

  \item For each $B$-compatible tuple, choose a rational
    reconstruction of the lower coefficient functions, that is,
    functions $\widehat c_{\ell j}(\Tau)\in \K(\Tau)$ for $\ell<e$ and
    $0\le j<m$, each with numerator and denominator degree at most
    $B$, whose values at the used specializations agree with the data
    $c_{\ell j}(a_i)$. From these, reconstruct a monic candidate
    \tightdisplay{
      \widehat h = x^e +
      \sum_{\ell=0}^{e-1}
      \left(
        \sum_{j=0}^{m-1}\widehat c_{\ell j}(\Tau)t^j
      \right)x^\ell
      \in \R.
    }
  \item Compute the right remainder $\rrem(g,\widehat h)$ in
    $\R$.
  \item If the remainder is zero, compute the cofactor $q$ with
    $g=q\widehat h$, recurse on $q$ and on $\widehat h$ until the
    current factor is completely factored, and return the resulting
    factorization.

  \item If no candidate survives certification for any admissible
    degree $e$ among the specializations, local candidate blocks,
    compatible tuples, and interpolation degree bound allowed by
    $\mathcal B_{\rm lift}$, return failure.
\end{enumerate}
\end{algorithm}

Algorithm~\ref{alg:lift} is the characteristic-zero analogue of
combining modular factors in commutative factorization. The difference
is that the local factors come from \emph{two} modular levels:
inert-prime factorization inside each specialization, and then
interpolation across several characteristic-zero specializations.  As
in Algorithm~\ref{alg:onespecialization}, Step~\textup{(2)} only
proposes candidates from the specialized factors.  It does not claim
that one fixed specialized factorization lists every right divisor of
degree $e$.

\section{A two-level algorithm for factorization over
  $\K(t)$}
\label{sec:alg}

We can now collect the previous stages into one procedure
over $\K(t)$.

The global budget for Algorithm~\ref{alg:full} is a finite collection
$\mathcal B=(\mathcal B_{\rm spec},\{\mathcal B_a\},\mathcal B_{\rm
  lift})$.  Here $\mathcal B_{\rm spec}$ is a finite range of central
specializations $a$ to test, $\mathcal B_a$ is the inert-prime and
recombination budget used by Algorithm~\ref{alg:onespecialization} for
that $a$, and $\mathcal B_{\rm lift}$ is the interpolation and
reconstruction budget used by Algorithm~\ref{alg:lift}. A larger
search means running the same procedure with larger finite
budget parameters; correctness of accepted factors is independent of
these choices.

\begin{algorithm}[H]
\caption{Two-level modular factorization in the quantum plane}
\label{alg:full}

\textbf{Input:} a nonzero skew polynomial $F_{\rm in}\in \R$ and a
finite global search budget $\mathcal B$.

\textbf{Output:} either a unit $\lambda\in\K(t)^\times$ and a complete
factorization
\[
  F_{\rm in}=\lambda h_1h_2\cdots h_r
\]
in $\K(t)[x;\sigma]$, or failure after $\mathcal B$ is exhausted.

\begin{enumerate}[label=\textup{(\arabic*)},leftmargin=2.2em]
\item Replace the input by a convenient monic left associate.  Clear
  denominators and remove scalar content to obtain
  $F\in\K[t][x;\sigma]$ and a unit $\gamma\in\K(t)^\times$ with
  $F=\gamma F_{\rm in}$.  Let $\ell=\operatorname{lc}_x(F)$ and set
  $f=F^\sharp=\ell^{-1}F$.  Record $\lambda=\gamma^{-1}\ell$, so that
  $F_{\rm in}=\lambda f$.  From this point on, run the factorization
  stages on the monic polynomial $f$.
  
\item Apply Algorithm~\ref{alg:central} to split $f$ into factors with
  central-irreducible bounds.
\item For each irreducible-bound factor $g$ of $f$:
  \begin{enumerate}[label=\textup{(\alph*)},leftmargin=2em]
  \item Check the conditions of Corollary~\ref{cor:maxdeg} (degree of
    irreducible bound) and Proposition~\ref{prop:normprefilter}
    (irreducible norm); if either test proves irreducibility, keep $g$
    and continue to the next current factor.
  \item otherwise choose good central specializations $a_1,a_2,\dots$
    from the finite range prescribed by $\mathcal B_{\rm spec}$;
  \item for each $a_i$, compute the specialized block size
    $d_{a_i}=\ind(A_{a_i})$ using cyclic-algebra methods over number
    fields;
  \item attempt to factor $g_{a_i}$ by
    Algorithm~\ref{alg:onespecialization} using its assigned budget
    $\mathcal B_{a_i}$;
  \item from any successful specialized factorizations, apply
    Algorithm~\ref{alg:lift} with budget $\mathcal B_{\rm lift}$ to
    reconstruct factors of $g$ over $\K(t)$.
  \end{enumerate}

\item If every current factor has been proved irreducible or
  completely factored, return $\lambda$ together with the resulting
  decomposition of $f$; otherwise return failure after the finite
  budget $\mathcal B$ has been exhausted.

\end{enumerate}
\end{algorithm}

\begin{theorem}
  Every factorization returned by Algorithm~\ref{alg:full} is correct:
  if the algorithm returns $\lambda,h_1,\ldots,h_r$, then
  \[
    F_{\rm in}=\lambda h_1h_2\cdots h_r
  \]
  in $\K(t)[x;\sigma]$, and every factor declared irreducible is
  either the degree-one factor $x$ or has been certified by
  one of the stated irreducibility tests.
\end{theorem}

\begin{proof}
  The normalization step records $F_{\rm in}=\lambda f$, where
  $\lambda\in\K(t)^\times$ is a unit and $f$ is monic.  All subsequent
  stages operate on this monic associate $f$. At the central stage,
  every non-$x$ factor is certified by right gcd computations, and the
  stripped copies of $x$ are irreducible by degree. The irreducibility
  declarations made before specialization are certified by verifying
  Corollary~\ref{cor:maxdeg} or Proposition~\ref{prop:normprefilter}.
  In the modular stage, specialized factors are certified by right
  division in the specialized rings, and lifted factors are certified
  by right division in the original ring $\K(t)[x;\sigma]$. Therefore
  the algorithm may fail to find a factorization, but any
  decomposition $f=h_1\cdots h_r$ that it returns is correct for the
  monic associate.  Multiplying on the left by the recorded unit
  $\lambda$ gives the asserted identity for $F_{\rm in}$. \qed
\end{proof}

No completeness is claimed for the modular reconstruction stage.
Algorithm~\ref{alg:full} searches for factors by trying bounded
families of specializations, finite primes, and recombinations. If no
candidate passes exact division in $\K(t)[x;\sigma]$, the algorithm
returns failure.

\section{Complexity and practical remarks}
\label{sec:complexity}

Only part of our method has been analyzed. A full bit-complexity
discussion must also include the coefficient height of the
input. Let $H$ denote a height bound for the coefficients of $f$ when
they are written in the basis $1,t,\dots,t^{m-1}$ over $\K(\Tau)$.

The costs identified are the following:
\begin{enumerate}[topsep=2pt,label=\textup{(C\arabic*)},leftmargin=2.4em]
\item Computing the bound and the reducible-bound stage are
  polynomial-time in the algebraic operations model, and should be
  effective in the bit model assuming commutative
  factorization over number fields (we do not do this analysis here).

\item The maximal-bound-degree test is negligible once the bound is
  known.

\item For each good specialization $a$, computing the specialized
  block size $d_a=\ind(A_a)$ invokes cyclic-algebra algorithms over
  number fields.  This step is effective, but the cost is not analyzed
  here.
  
\item For fixed specialization $a$ and inert prime $\q$, the
  factorization of the finite-field image of degree $n$ is handled and
  fully analyzed in \citet{CarLeb17} and \citet{Gie98}. This is the
  fastest stage of the algorithm.

\item Reconstruction over one specialization uses Chinese remaindering
  in the Dedekind domain $\OLa$ and linear algebra in an integral
  basis of $\L_a$.

\item Lifting from several specializations uses rational interpolation over $\K$ in the single central variable $\Tau$.

\end{enumerate}

The difficult part is the search. In the heuristic
algorithms above, the search is deliberately split into finite budgets:
\begin{enumerate}[label=\textup{(S\arabic*)},leftmargin=2.4em]
\item the central-specialization budget $\mathcal B_{\rm spec}$, which
  bounds the values $a$ tested;
\item the inert-prime budget inside each specialization, which bounds
  the prime norms and the number of finite-field reductions;
\item the modular block-recombination budget, which bounds the number
  of candidate block collections tried before Chinese remaindering and
  rational reconstruction over $\L_a$;
\item the lifting budget, including the interpolation degree bound $B$
  and the number of $B$-compatible tuples tried across
  specializations.
\end{enumerate}
These finite budgets make ``failure'' a possible outcome of the chosen
search range. They do not enumerate the many, possibly infinitely
many, right divisors of an irreducible-bound input.  They test only
finitely many candidate tuples. Certification, not exhaustion of the
search space, is what makes accepted factors correct.  We make no
complexity claim for the irreducible-bound stage over $\K(t)$.

The practical appeal is that the centre is a univariate PID, central
specializations $\Tau\mapsto a$ are easy to sample, and much of the
noncommutative work can be pushed into finite-field images where fast
algorithms already exist.

\section{Historical context and related work}
\label{sec:history}

Our strategy combines four earlier viewpoints: the centre-and-bound
methods of \citet{GieZha03} and \citet{GomLob19}, the root-of-unity
quantum-plane studied by \citet{CouPri06}, norm methods in
finite-order skew-polynomial rings \citep{PumTho22}, and algorithms
for cyclic and central simple algebras over number fields
\citep{Han07,IvaRon12}. A purely finite-field reduction loses some of
the structure attached to an irreducible bound. We therefore first
specialize in characteristic zero, where this structure is still
available, and only then reduce modulo good primes to obtain fast
finite-field factorizations.

The polynomial quantum plane can also be treated by the noncommutative
subsystem \textsc{Plural} of \textsc{Singular} \citep{Lev06}. That
system provides PBW-type normal forms, noncommutative Gr\"obner bases,
elimination, and related algorithms for many noncommutative polynomial
algebras.

The \texttt{ncfactor.lib} library of \citet{LevHei18} implements
ansatz-driven factorization methods for several such algebras,
building on finite-factorization-domain results of \citet{BelHei17}
for $G$-algebras and related filtered algebras.  This work is
therefore important related software, especially for polynomial
quantum-plane computations.  The problem considered here is different
in two respects.  First, we factor in the localized Ore ring
$\K(t)[x;\sigma]$, not only in the polynomial $G$-algebra
$\K\langle t,x\rangle/(xt-\omega tx)$, so denominators in $t$ may
occur in factors.  Second, the algorithm exploits the root-of-unity
centre $\K(t^m)[x^m]$, through the additional structure of central
bounds, right-gcd splitting, irreducible-bound cyclic algebras, good
specialization, inert-prime finite-field reduction, modular
recombination, and interpolation.

\section{What is presently algorithmic over $\K(t)$?}
\label{sec:algstatus}

We now separate the parts that are currently computable from the parts
that are only searchable or remain open over $\K(t)$.

\begin{proposition}
\label{prop:algstatus}
The following tasks are computable (without claim of efficiency):
\begin{enumerate}[label=\textup{(\alph*)},leftmargin=2.2em]
\item Normalization of the input to a primitive integral representative
  and a recorded monic associate, computation of the monic bound
  $\varphi$, commutative factorization of $\varphi$ in
  $\C=\K(\Tau)[\Chi]$, and the rough factorization of the working monic
  polynomial $f$ into central-irreducible pieces by right gcd
  computations.

\item The maximal-bound-degree test of Corollary~\ref{cor:maxdeg} and
  the norm-based irreducibility test of Proposition~\ref{prop:normprefilter}.

\item For each chosen good specialization $\Tau\mapsto a$ and each
  finite set of chosen good inert primes $\q_1,\ldots,\q_r$, the
  construction of the specialized skew ring and of its finite-field
  images, finite-field skew factorization, recombination of modular
  factors, Chinese remaindering over the primes
  $\p_i\subset \OLa$ above the $\q_i$, rational reconstruction in
  $\L_a$, and certification of any reconstructed factor by right
  division.

\item Rational interpolation of candidate coefficients from finitely
  many characteristic-zero specializations, followed by verification
  in the original ring.
\end{enumerate}
\end{proposition}

\begin{proof}[Discussion]
  Parts~\textup{(a)}--\textup{(d)} are immediate from the constructions
  in Sections~\ref{sec:centraldecomp}--\ref{sec:lift}, together with
  the cited commutative and finite-field skew-factorization routines.
  For part~\textup{(c)}, if $\q_i$ is inert in $\L_a/\K$, then there is
  a unique prime $\p_i\subset \OLa$ above $\q_i$.  For distinct
  $\q_i$ the ideals $\p_i$ are pairwise coprime, and the Chinese
  remainder theorem in the Dedekind domain $\OLa$ gives an effective
  isomorphism
  \[
    \OLa/\prod_i \p_i \simeq \prod_i \OLa/\p_i .
  \]
  Using a fixed integral basis of $\OLa$, this reconstructs candidate
  integral coefficient data coefficientwise from the finite-field
  images.  After common denominator clearing, rational reconstruction
  in $\L_a$ produces candidate coefficients in the specialized ring
  $\L_a[x;\sigma_a]$.  These steps only propose possible factors.
  Every proposed factor is checked afterwards by exact division in
  $\L_a[x;\sigma_a]$, and any lifted factor is checked by exact
  division in $\R$.  \qed
\end{proof}

Because $\K(t)$ is countable and effectively presented, we can
enumerate all monic candidates of smaller $x$-degree in
$\K(t)[x;\sigma]$ and test them by right division. Thus a search for a
proper factor certifies reducibility when it succeeds. To obtain a
full decision procedure, we also need a way to certify irreducibility.

\subsection{The remaining difficulties over $\K(t)$}

After the reductions above, the part not covered by the certified
tests is the irreducible-bound case over the generic field $\K(\Tau)$,
as made precise in the following proposition.

\begin{proposition}
  \label{prop:conditionaldec}
  Assume that for every monic irreducible-bound input
  $g\in \K(t)[x;\sigma]$ with monic bound $\pi\in \K(\Tau)[\Chi]$,
  $\pi\neq \Chi$, there is an algorithm which constructs the
  finite-dimensional algebra
  \[
    E_g:=\operatorname{End}_{A_\pi}(M_g), \qquad
    \F_\pi=\K(\Tau)[\Chi]/(\pi), \qquad
    A_\pi=\R/\R\pi, \qquad M_g=\R/\R g,
  \]
  and either proves that $E_g$ is a division algebra or returns a
  nontrivial idempotent of $E_g$. Then full factorization in
  $\K(t)[x;\sigma]$ is decidable.
\end{proposition}

\begin{proof}[Sketch]
  Apply Algorithm~\ref{alg:central} to reduce to irreducible-bound
  pieces. Let $g$ be one such piece, with bound $\pi$. Since
  $\pi\in \R g$, the containment $\R\pi\subseteq \R g$ gives
  $M_g=\R/\R g$ the structure of a left $A_\pi$-module.  Write
  $A_\pi\cong M_r(D)$, where $D$ is a central division
  $\F_\pi$-algebra of degree $d$.  If $S$ is a simple left
  $A_\pi$-module, then
  \[
    M_g\cong S^\ell, \qquad \ell=\frac{\deg_x g}{s d}, \qquad
    s=\deg_\Chi \pi,
  \]
  and therefore
  $E_g=\operatorname{End}_{A_\pi}(M_g)\cong M_\ell(D^{\rm op})$.  Thus
  $g$ is irreducible if and only if $\ell=1$, equivalently if and only
  if $E_g$ is a division algebra.

  If the assumed algorithm proves that $E_g$ is a division algebra,
  then $g$ is irreducible. Otherwise it returns a nontrivial
  idempotent $e\in E_g$. The image $eM_g$ is a proper nonzero
  $A_\pi$-submodule of $M_g$. Taking its inverse image under
  $\R\to M_g$ gives an intermediate left ideal
  $\R g \subsetneq I \subsetneq \R$.  Since $\R$ is a left principal
  ideal domain, write $I=\R h$. Then $g\in \R h$, so $g=q h$ for some
  $q\in \R$, and $h$ is a proper right factor of $g$. Right
  division verifies the factorization.  Recursing on the 
  cofactors gives a complete factorization.  \qed
\end{proof}

Proposition~\ref{prop:conditionaldec} isolates the difficulty in our
problem.  The unresolved object is not merely the central simple
algebra $A_\pi$ attached to the central prime $\pi$, but the
$A_\pi$-module $M_g=\R/\R g$, or equivalently its eigenring
$E_g=\operatorname{End}_{A_\pi}(M_g)$. A zero divisor in $A_\pi$ alone
does not necessarily split the particular polynomial $g$.

\section{SageMath implementation and computational examples}
\label{sec:implementation}

We have made a SageMath 10.8 prototype of Algorithm~\ref{alg:full} and
its subroutines \citep{sagemath}.  The code and sample runs are
available for experimentation.%
\footnote{\url{https://cs.uwaterloo.ca/~mwg/archive/quantum_plane_factorization}}
It implements arithmetic in $\K(t)[x;\sigma]$, right division,
right gcds, central-bound computation, central splitting, the two
irreducibility filters, good specialization, inert-prime reduction,
finite-field skew factorization using the Caruso--Le Borgne backend,
recombination, interpolation in $\Tau$, and final certification.
A small sparse-search fallback is used only for deliberately
constructed examples such as $(x^e-t)^2$ or $\Chi^e-\Tau$, and does
not replace the missing generic irreducible-bound algorithm of
Section~\ref{sec:algstatus}.

Table~\ref{tab:sage-timings} below summarizes one run of the example suite
on an Apple Mac Mini M2 Pro with 32GB RAM.  The examples exercise
central splitting, the degree-bound filter, modular/sparse
certification, and irreducible-central inputs.  The timings are
implementation data, not complexity claims.

\begin{table}[ht]
\centering
\small
\setlength{\tabcolsep}{4pt}
\caption{Timings for the SageMath example suite.  The column ``bound''
lists the degrees of the irreducible factors of the central bound
$\varphi\in\K(\Tau)[\Chi]$.  The column ``output'' lists the
$x$-degrees of the skew factors; $d\times r$ means $r$ factors of
degree $d$.  Here C means central
factorization only, C+DB means central factorization plus degree-bound
irreducibility tests, M/S means modular attempt plus sparse
certification, and CI/S means factorization of an irreducible central
polynomial by sparse certification.}
\label{tab:sage-timings}
\begingroup
\footnotesize
\setlength{\tabcolsep}{3pt}
\renewcommand{\arraystretch}{0.92}
\begin{adjustbox}{max width=\textwidth}
\begin{tabular}{@{}c l r r r l l S[table-format=3.3]@{}}
\toprule
ID & stage & {$m$} & {$\deg_x f$} & {$\deg_\Chi\varphi$}
   & bound & output & {time (s)} \\
\midrule
C1 & C    & 2 & 2  & 2  & $1\times 2$                 & $1\times 2$           & 0.012 \\
C2 & C    & 2 & 5  & 5  & $1\times 5$                 & $1\times 5$           & 0.069 \\
C3 & C    & 3 & 5  & 5  & $1\times 5$                 & $1\times 5$           & 0.396 \\
I1 & M/S  & 3 & 4  & 2  & $2$                         & $2\times 2$           & 0.850 \\
I2 & M/S  & 3 & 4  & 2  & $2$                         & $2\times 2$           & 0.403 \\
B1 & CI/S & 3 & 6  & 2  & $2$                         & $2\times 3$           & 0.031 \\
S1 & C    & 3 & 8  & 8  & $1\times 8$                 & $1\times 8$           & 2.145 \\
S2 & C+DB & 3 & 12 & 12 & $1\times 6,\;2,\;4$          & $1\times 6,\;2,\;4$   & 148.660 \\
S3 & C    & 5 & 6  & 6  & $1\times 6$                 & $1\times 6$           & 0.972 \\
S4 & M/S  & 3 & 8  & 4  & $4$                         & $4\times 2$           & 599.347 \\
S5 & CI/S & 3 & 15 & 5  & $5$                         & $5\times 3$           & 0.049 \\
\midrule
\multicolumn{7}{r}{Total wall time for the selected examples}
  & 752.934 \\
\bottomrule
\end{tabular}
\end{adjustbox}
\endgroup
\end{table}

All rows ended with a product check in $\K(t)[x;\sigma]$. 
The examples do not prove completeness for the irreducible-bound
stage, but they show that the algorithmic components in
Proposition~\ref{prop:algstatus} assemble into a working 
implementation.  

\section{Decidability over $\QQbar(t)$ in the exact algebraic model}
\label{sec:algclosure}

We next consider the same factorization problem after enlarging the
field of constants from $\K$ to the algebraic closure $\QQbar$.  The
analogous situation over $\CC$ will be discussed separately in
Subsection~\ref{subsec:bssfact}. In this setting the algebra
associated with an irreducible bound always splits, so the remaining
task is to find the factors effectively. Set
\tightdisplay{
  \Ralg=\QQbar(t)[x;\sigma], \qquad \sigma(t)=\omega t,
}
with centre
\tightdisplay{
  \Calg=\QQbar(\Tau)[\Chi], \qquad \Tau=t^m,\ \Chi=x^m.
}

By the \emph{standard exact algebraic model over $\QQbar$} we mean the
usual symbolic computation model in which every input is given over
some explicit finite extension $E_0/\QQ\subset \QQbar$, represented by
a primitive element, and computations are then carried out in
$E_0$ and in further finite extensions obtained by adjoining roots as
needed.  In this sense $\QQbar$ is handled as the directed union of
the finite extensions of $\QQ$. See \citet{Ste10} for an effective
algorithmic approach to computing in algebraic closures.

\subsection{Irreducible-bound pieces over algebraically closed
  constants}

Over an algebraically closed field of constants the simple algebra
attached to an irreducible bound becomes just a matrix algebra over
its centre.  We say that such an algebra \emph{splits} over a field
$\F$, meaning that it is isomorphic to a full matrix algebra $M_r(\F)$
for some $r$.  After moving to the algebraic closure, the difficulty
coming from the associated algebra disappears.  The irreducible
factors then have the expected degree $s$ in Fact~\ref{fact:degrees}.

\begin{theorem}
  \label{thm:qbarsplit}
  Let $g\in \Ralg$ be monic with irreducible bound $\pi\in \Calg$,
  $\pi\neq \Chi$, and write $s=\deg_\Chi(\pi)$, $\Falg=\Calg/(\pi)$,
  and $A=\Ralg/\Ralg\pi$.  Then $A$ is a split central simple algebra
  over $\Falg$. Thus, every irreducible factor of $g$ in $\Ralg$ has
  degree $s$ in $x$. In particular, $g$ is irreducible if and
  only if $\deg_x(g)=s$.
\end{theorem}

\begin{proof}
  The field $\Falg$ is the function field of a curve over the
  algebraically closed field $\QQbar$. By Tsen's theorem, $\Falg$ is a
  $C_1$-field, and since $C_1$-fields have trivial Brauer group, we
  get $\Br(\Falg)=0$.  See \citep[Theorem~6.2.8 and
  Proposition~6.2.3]{GilSza06}.  Since $A$ is a finite-dimensional
  central simple $\Falg$-algebra, it follows that $A$ is split, i.e.,
  $A\cong M_m(\Falg)$.  Thus the index in Fact~\ref{fact:degrees} is
  $d=1$, and every irreducible factor degree is~$sd=s$. \qed
\end{proof}

\subsection{A complete factorization algorithm over $\QQbar(t)$}

The splitting theorem above immediately turns the irreducible-bound case
into a (very inefficient) search problem.

\begin{theorem}
  \label{thm:qbardec}
  In the exact algebraic model over $\QQbar$, factorization in
  $\QQbar(t)[x;\sigma]$ is decidable. More precisely, there is an
  algorithm that returns a complete factorization of any monic
  input in $\QQbar(t)[x;\sigma]$.
\end{theorem}

\begin{proof}
  Let $f\in \QQbar(t)[x;\sigma]$ be monic. Choose a finite extension
  $E_0/\QQ$ inside $\QQbar$ that contains $\K$ and all coefficients of
  $f$. Let $\varphi$ be the monic bound of $f$. In the algebraic model
  over $\QQbar$, factor $\varphi$ completely in
  $\Calg=\QQbar(\Tau)[\Chi]$. Since $\varphi$ has finitely many
  coefficients, the coefficients of its central irreducible factors
  lie in some finite extension $E_1/\QQ$ obtained from $E_0$ by
  adjoining finitely many algebraic constants. Replacing $E_0$ by
  $E_1$, apply the same central decomposition as
  Algorithm~\ref{alg:central} over $E_1(t)[x;\sigma]$ to split $f$
  into copies of $x$ and factors whose bounds are irreducible in
  $\Calg$. It therefore suffices to treat one such irreducible-bound
  piece $g$ with bound $\pi$.
  
  Set $s=\deg_\Chi(\pi)$. If $\deg_x(g)=s$, then $g$ is irreducible by
  Theorem~\ref{thm:qbarsplit}. Assume therefore that $\deg_x(g)>s$.
  Then Theorem~\ref{thm:qbarsplit} implies that $g$ has a nontrivial
  right factor of $x$-degree $s$.
  
  We now dovetail over finite extensions $\E/\QQ$ containing $E_1$,
  represented by primitive elements, and over tuples in $\E(t)$ of
  rational functions that can occur as the lower coefficients of a
  monic degree-$s$ skew polynomial
  \[
    h=x^s+h_{s-1}x^{s-1}+\cdots+h_0.
  \]
  For each candidate $h$, compute the right remainder
  $\rrem(g,h)$. Because an actual right factor of degree $s$ exists
  and uses only finitely many algebraic coefficients, it lies in
  $\E(t)[x;\sigma]$ for some finite extension $\E/\QQ$ containing
  $E_1$, so enumeration eventually finds it. At that stage the right
  remainder is zero and we have found a factor of~$g$.

  Compute the cofactor by right division and recurse on both
  pieces. Since the $x$-degree drops at each recursive step, the
  process terminates and returns a complete factorization of $f$. \qed
\end{proof}

Theorem~\ref{thm:qbardec} is a decidability result, not an efficiency
theorem. After passing to the algebraic closure of the constants, the
associated algebra splits. Thus factors of the expected degree exist,
though the algorithm may still have to search through many candidates
to find them.

\subsection{A brief comment on factoring in the BSS model over $\CC$}
\label{subsec:bssfact}

The preceding discussion can also be viewed in the Blum--Shub--Smale
model over $\CC$ \citep{BluCuc98}.  Replacing $\QQbar$ by $\CC$ does
not change the algebraic splitting statement;
Theorem~\ref{thm:qbarsplit} still shows that, for an irreducible-bound
input, reducibility is detected by comparing the degree in $x$ of the
input with the $\Chi$-degree of its bound.  Thus, once such a piece
and its bound have been computed, this particular reducibility test is
available in the bare BSS model.

What is not built into the bare BSS model is a uniform operation for
choosing roots of auxiliary polynomials.  This is the only further
point we need here.  The constructive proof of
Theorem~\ref{thm:qbardec} uses algebraic choices when it passes to
finite extensions of the constant field.  In a BSS model augmented
with such root choices, the same algebraic argument applies over
$\CC$.  Without that augmentation, we make no separate BSS
decidability claim.  See \citet{Bra05} for a discussion of this
root-choice issue.

\section*{Acknowledgements}
The author would like to thank the anonymous referees for their input.  The author acknowledges support of the Natural Sciences and Engineering Research Council (NSERC), Canada.

\patchcmd{\thebibliography}
  {\list}
  {\vspace{-3.25em}\list}
  {}
  {}
  
\begingroup
\renewcommand{\bibfont}{\small}
\setlength{\bibsep}{0pt}
\bibliographystyle{plainnat}

\renewcommand{\doi}[1]{%
  \href{https://doi.org/#1}{doi:\footnotesize\nolinkurl{#1}}%
}

\bibliography{qdilfact}

@article{Ore33,
  author  = {Ore, Oystein},
  title   = {Theory of Non-Commutative Polynomials},
  journal = {Annals of Mathematics},
  volume  = {34},
  number  = {3},
  pages   = {480--508},
  year    = {1933}
}

@article{CouPri06,
  author  = {Coulibaly, Romain and Price, Kenneth L.},
  title   = {Factorization in Quantum Planes},
  journal = {Missouri J. Mathematical Sciences},
  volume  = {18},
  number  = {3},
  pages   = {197--205},
  year    = {2006},
  doi     = {10.35834/2006/1803197}
}

@article{Gie98,
  author  = {Giesbrecht, Mark},
  title   = {Factoring in Skew-Polynomial Rings over Finite Fields},
  journal = {J. Symb. Comp.},
  volume  = {26},
  number  = {4},
  pages   = {463--486},
  year    = {1998},
  doi     = {10.1006/jsco.1998.0224}
}

@inproceedings{GieZha03,
  author    = {Giesbrecht, Mark and Zhang, Yang},
  title     = {Factoring and Decomposing {O}re Polynomials over {$\mathbb{F}_q(t)$}},
  booktitle = {Proceedings of the 2003 International Symposium on Symbolic and Algebraic Computation (ISSAC 2003)},
  pages     = {127--134},
  publisher = {ACM},
  address   = {New York},
  year      = {2003},
  doi       = {10.1145/860854.860888}
}

@article{CarLeb17,
  author  = {Caruso, Xavier and Le Borgne, J{\'e}r{\'e}my},
  title   = {A New Faster Algorithm for Factoring Skew Polynomials over Finite Fields},
  journal = {J. Symb. Comp.},
  volume  = {79},
  number  = {2},
  pages   = {411--443},
  year    = {2017},
  doi     = {10.1016/j.jsc.2016.02.016}
}

@article{GomLob19,
  author  = {{G{\'o}mez-Torrecillas}, Jos{\'e} and Lobillo, F. J. and Navarro, Gabriel},
  title   = {Computing the Bound of an {O}re Polynomial. {Applications} to Factorization},
  journal = {J. Symb. Comp.},
  volume  = {92},
  pages   = {269--297},
  year    = {2019},
  doi     = {10.1016/j.jsc.2018.04.018}
}

@inproceedings{Han07,
  author    = {Hanke, Timo},
  title     = {The Isomorphism Problem for Cyclic Algebras and an Application},
  booktitle = {Proceedings of the 2007 International Symposium on Symbolic and Algebraic Computation (ISSAC 2007)},
  pages     = {181--186},
  publisher = {ACM},
  address   = {New York},
  year      = {2007},
  doi       = {10.1145/1277548.1277574}
}

@article{IvaRon12,
  author  = {Ivanyos, G{\'a}bor and R{\'o}nyai, Lajos and Schicho, Josef},
  title   = {Splitting Full Matrix Algebras over Algebraic Number Fields},
  journal = {J. Algebra},
  volume  = {354},
  number  = {1},
  pages   = {211--223},
  year    = {2012},
  doi     = {10.1016/j.jalgebra.2012.01.008}
}

@article{PumTho22,
  author  = {Pumpl{\"u}n, Susanne and Thompson, Daniel},
  title   = {The Norm of a Skew Polynomial},
  journal = {Algebras and Representation Theory},
  volume  = {25},
  pages   = {869--887},
  year    = {2022},
  doi     = {10.1007/s10468-021-10051-z}
}

@book{BluCuc98,
  author    = {Blum, Lenore and Cucker, Felipe and Shub, Michael and Smale, Steve},
  title     = {Complexity and Real Computation},
  publisher = {Springer},
  address   = {New York},
  year      = {1998},
  doi       = {10.1007/978-1-4612-0701-6}
}

@book{GilSza06,
  author    = {Gille, Philippe and Szamuely, Tam{\'a}s},
  title     = {Central Simple Algebras and Galois Cohomology},
  publisher = {Cambridge University Press},
  address   = {Cambridge},
  year      = {2006}
}

@inproceedings{Bra05,
  author    = {Braverman, Mark},
  title     = {On the Complexity of Real Functions},
  booktitle = {Proceedings of the 46th Annual {IEEE} Symposium on Foundations of Computer Science},
  pages     = {155--164},
  publisher = {IEEE Computer Society},
  address   = {Los Alamitos, CA},
  year      = {2005},
  doi       = {10.1109/SFCS.2005.58}
}

@Book{Jac43,
  author = "N. Jacobson",
  title = "The Theory of Rings",
  publisher = "American Math. Soc.",
  address = "New York",
  year = 1943
}

@preamble{"\newcommand{\Gathen}{\relax}"}

@string{Gathen = "\Gathen{von zur Gathen}, J."}

@book{GatGer13,
title = "Modern Computer Algebra",
edition = "3",
author = Gathen # " and J. Gerhard",
publisher = "Cambridge University Press",
year = 2013,
place = "Cambridge, UK"
}

@article{Ste10,
title = {Computing with algebraically closed fields},
volume = {45},
doi = {10.1016/j.jsc.2009.09.005},
number = {3},
journal = {J. Symb. Comp.},
author = {Steel, Allan K.},
year = {2010},
pages = {342--372},
}

@inproceedings{Lev06,
  author = {Levandovskyy, Viktor},
  title = {Plural, a Non-commutative Extension of Singular:
           Past, Present and Future},
  booktitle = {Mathematical Software -- ICMS 2006},
  series = {Lecture Notes in Computer Science},
  volume = {4151},
  pages = {144--157},
  publisher = {Springer},
  year = {2006},
  doi = {10.1007/11832225_13}
}

@article{LevHei18,
  author = {Levandovskyy, Viktor and Heinle, Albert},
  title = {A Factorization Algorithm for {$G$}-Algebras and Its Applications},
  journal = {J. Symb. Comp.},
  volume = 85,
  pages = {188--205},
  year = 2018,
  doi = {10.1016/j.jsc.2017.06.005}
}

@manual{sagemath,
  author = {{The Sage Developers}},
  title  = {{S}ageMath, the {S}age {M}athematics {S}oftware {S}ystem},
  url    = {https://www.sagemath.org},
  version = {10.8},
  year   = {2025}
}

@article{BroPum21,
  author  = {Brown, C. and Pumpl{\"u}n, S.},
  title   = {Irreducible skew polynomials over domains},
  journal = {Analele {\c S}tiin{\c t}ifice ale Universit{\u a}{\c t}ii Ovidius Constan{\c t}a. Seria Matematic{\u a}},
  volume  = {29},
  number  = {3},
  pages   = {75--89},
  year    = {2021},
  doi     = {10.2478/auom-2021-0035}
}

@article{BelHei17,
  author  = {Bell, Jason P. and Heinle, Albert and Levandovskyy, Viktor},
  title   = {On noncommutative finite factorization domains},
  journal = {Transactions of the American Mathematical Society},
  volume  = {369},
  number  = {4},
  pages   = {2675--2695},
  year    = {2017},
  doi     = {10.1090/tran/6727}
}
\endgroup

\end{document}